\documentclass[review]{elsarticle}

\usepackage{lineno,hyperref}
\usepackage{amsmath, amssymb, amsthm}
\usepackage{soul}
\usepackage{mathtools}
\usepackage{graphics}
\usepackage{graphicx}
\usepackage{epsfig}
\usepackage{dblfloatfix}
\usepackage{float}

\newtheorem{theorem}{\bf Theorem}
\newtheorem{lemma}{\bf Lemma}

\newtheorem{example}{\bf Example}
\newtheorem{remark}{\bf Remark}

\usepackage{epstopdf}
\usepackage{url}
\usepackage{xcolor}
\usepackage{hyperref}
\date{}
\usepackage{fullpage}
\usepackage{setspace}
\usepackage{lscape}

\newtheorem{df}{Definition}[section]

\newtheorem{corollary}{Corollary}[section]

\journal{Discrete Math.}

\begin{document}

\begin{frontmatter}

\title{Cyclic codes over a non-chain ring $R_{e,q}$ and their application to LCD codes}


\author{Habibul Islam}
\address{Department of Mathematics, Indian Institute of Technology Patna, India}
\ead{habibul.pma17@iitp.ac.in}

\author{Edgar Mart\'inez-Moro}
\address{Institute of Mathematics, University of Valladolid, Spain}
\ead{edgar.martinez@uva.es}

\author{Om Prakash\corref{mycorrespondingauthor}}
\address{Department of Mathematics, Indian Institute of Technology Patna, India}
\cortext[mycorrespondingauthor]{Corresponding author}
\ead{om@iitp.ac.in}

\begin{abstract}
Let $\mathbb{F}_q$ be a finite field of order $q$, a prime power integer such that $q=et+1$ where $t\geq 1,e\geq 2$ are integers. In this paper, we study cyclic codes of length $n$ over a non-chain ring $R_{e,q}=\mathbb{F}_q[u]/\langle u^e-1\rangle$. We define a Gray map $\varphi$ and obtain many { maximum-distance-separable} (MDS) and optimal $\mathbb{F}_q$-linear codes from the Gray images of cyclic codes. Under certain conditions we determine { linear complementary dual} (LCD) codes of length $n$ when $\gcd(n,q)\neq 1$ and $\gcd(n,q)= 1$, respectively. It is proved that { a} cyclic code $\mathcal{C}$ of length $n$ is an LCD code if and only if its Gray image $\varphi(\mathcal{C})$ is an LCD code of length $4n$ over $\mathbb{F}_q$. Among others, we present the conditions for existence of free and non-free LCD codes. Moreover, we obtain many optimal LCD codes as the Gray images of non-free LCD codes over $R_{e,q}$.
\end{abstract}

\begin{keyword}
Cyclic code \sep LCD code \sep Optimal code \sep MDS code \sep self-dual code \sep Gray map.
\MSC[2010] 94B15\sep 94B05\sep 94B60.
\end{keyword}

\end{frontmatter}


\section{Introduction}
Cyclic codes are one of the most important classes of linear codes and they possess rich algebraic properties. { A} cyclic code of length $n$ over a ring (or field) $R$ is defined as an $R$-submodule (resp. subspace) of $R^n$ closed under the cyclic shift. Over a finite field $\mathbb{F}_q$ they can be easily classified via ideals of the ring $\mathbb{F}_q[x]/\langle x^n-1\rangle$ but usually their structure is non-trivial for the case of finite rings. The algebraic structure such as generator polynomials, generator matrices are useful to obtain their parameters like rank or minimum distances. In the field case the standard notation $[n,k,d]$ is used to represent a linear code of length $n$, dimension $k$ and distance $d$. { There are many well-known theoretical relations among these parameters, and some of them are Singleton bound, Plotkin bound, Gilbert-Varshamov bound, Hamming bound and Griesmer bound, etc}. A code { attaining} any of these bounds is called maximum-distance-separable (MDS) for that bound. In 1997, Kanwar and Lopez-Permouth \cite{Kanwar97} determined the structure of cyclic codes over $\mathbb{Z}_{p^m}$ while Pless et al. \cite{Pless97} obtained self-dual cyclic codes over $\mathbb{Z}_4$.  Later, in 1999, Bonnecaze and Udaya \cite{Bonnecaze99} studied cyclic codes over the ring $\mathbb{F}_2+u\mathbb{F}_2,u^2=0$. Also, Blackford \cite{Blackford03} classified all cyclic codes of length $2n$ ($n$ is odd integer) over $\mathbb{Z}_4$. After that,  Abualrub et al. \cite{Abualrub04} discussed these codes of length $2^e$ over $\mathbb{Z}_4$. Also, in a more general setup, cyclic and negacyclic codes over finite chain rings are extensively studied in \cite{Dinh04}. Later, many papers contributed significant tools and results to obtain the structural properties of these codes over finite rings, see \cite{Abualrub07,Bayram14,Gao15,Islam19,Islam18,Islam21,Zhu10}.

Linear complementary dual (shortly LCD) codes were introduced by Massey \cite{Massey92}, in 1992. An LCD code is defined { as} a linear code having a trivial intersection with its dual code. These codes have gained serious attention due to their recent successful application in cryptography \cite{Carlet16}. In 2016, Carlet and Guilley \cite{Carlet16} shown that LCD codes with possibly large minimum distances prevent the resistance against side-channel attacks in a cryptosystem. In addition, these codes have a much simpler nearest-codeword decoding algorithm than linear codes, and { are} widely used in communications systems, data storage, and consumer electronics. In 1994, Yang and Massey \cite{Yang94} completely classified cyclic LCD codes over finite fields. Later,  Sendrier \cite{Sendrier04} shown that LCD codes attain the Gilbert-Varshamov bound. In 2009 Esmaeili and Yari \cite{Esmaeili09} studied  quasi-cyclic LCD codes. Later, Li et al. \cite{Li17} obtained many optimal and good LCD codes over finite fields. Also, Dougherty et al. \cite{Dougherty17} discussed LCD codes along with the linear programming bounds. Further, Galvez et al. \cite{Galvez18} determined many bounds on LCD codes over $\mathbb{F}_2$.

Recently, LCD codes over finite chain rings are extensively studied in \cite{Durgun20,Liu15,Liu19}. They have determined necessary and sufficient conditions of a linear code to be an LCD code and show the existence of asymptotically good LCD codes over finite chain rings. Therefore, it is still open to investigate the structure of LCD codes over finite commutative non-chain rings. Very recently, these codes have been explored over a non-chain ring $\mathbb{F}_q+u\mathbb{F}_q+v\mathbb{F}_q$ by Yadav et al. \cite{Yadav21}. Here, we consider a class of finite commutative non-chain rings with unity $R_{e,q}=\mathbb{F}_q[u]/\langle u^e-1\rangle$ where $q=et+1$ and $e\geq 2,t\geq 1$ are integers. Then  $R_{e,q}$ is a semi-local Frobenius ring of order $q^e$. The main concern of this paper is to obtain the structure of cyclic and LCD codes of length $n$ over $R_{e,q}$ and to construct optimal linear (or LCD) codes from them.  With the help of a new Gray map $\varphi$, we derive many optimal and MDS codes from cyclic codes over $R_{e,q}$. Further, we obtain necessary and sufficient conditions of cyclic codes to LCD codes. Finally, many optimal LCD codes are presented to validate our obtained results.

\section{Preliminary}\label{sec 2}
Let $p$ be an odd prime and $q=p^m$ such that $q=et+1$ where $t,m\geq 1,e\geq 2$ are integers. Let $\mathbb{F}_{q}$  be a field of order $q$ and denote $R_{e,q}=\mathbb{F}_{q}[u]/\langle u^{e}-1\rangle$. Then it is easy to check that $R_{e,q}$ is a finite commutative semi-local, Frobenius and non-chain ring with unity and characteristic $p$. Also $R_{e,q}$ can be written as  $R_{e,q}=\mathbb{F}_{q}+u\mathbb{F}_{q}+u^{2}\mathbb{F}_{p}+\dots+u^{e-1}\mathbb{F}_q$, where $u^{e}=1$, thus an arbitrary element $r\in R_{e,q}$ can be written as $r=a_0+ua_1+u^2a_2+\dots +u^{e-1}a_{e-1}$, where $a_i\in \mathbb{F}_{q}$ for $0\leq i\leq e-1$. Recall that a nonempty subset $\mathcal{C}\subseteq R^{n}_{e,q}$ is said to be a \textit{linear code} of length $n$ if it is an $R_{e,q}$-submodule of $R^{n}_{e,q}$ and elements in  $\mathcal{C}$ are called codewords. Also, the rank of a code $\mathcal{C}$ is the minimum number of generators for $\mathcal{C}$ and the free rank is the rank of $\mathcal{C}$ when it is a free module over $R_{e,q}$. A typical linear code of length $n$ and rank $k$ is denoted by $[n,k]$. The inner product between two vectors $x=(x_{0},x_{1},\cdots,x_{n-1}), y=(y_{0},y_{1},\cdots,y_{n-1})\in R_{e,q}^n$ is defined as $x\cdot y=\sum_{i=0}^{n-1}x_{i}y_{i}$. The orthogonal of a linear code $\mathcal{C}$ is denoted by $\mathcal{C}^{\perp}$ and defined by $\mathcal{C}^{\perp}=\{x\in R_{e,q}^{n}\mid x\cdot y=0,~\forall~y\in \mathcal{C} \}$. Note that the orthogonal code $\mathcal{C}^{\perp}$ is also a linear code of length $n$ over $R_{e,q}$.  The linear code $\mathcal{C}$ is said to be \textit{self-orthogonal} if $\mathcal{C}\subseteq \mathcal{C}^{\perp}$ and \textit{self-dual} if $\mathcal{C}=\mathcal{C}^{\perp}$.\\
Since $e\mid (q-1)$, the polynomial $u^e-1$ uniquely splits into linear factors over the ring $\mathbb{F}_q[u]$. Suppose
\begin{align*}
    u^{e}-1=\prod_{i=1}^e(u-\alpha_i) \in \mathbb{F}_{q}[u].
\end{align*}
Let $G_{i}=u-\alpha_{i}$ and $\widehat{G}_{i}=\frac{(u^{e}-1)}{G_{i}}$, for $1\leq i\leq e$. Since $\gcd(G_{i}, \widehat{G}_{i})=1$, there exist $z_{i},h_{i}\in \mathbb{F}_{q}[u]$ such that $z_{i}G_{i}+h_{i}\widehat{G}_{i}=1$, for $1\leq i\leq e$.
Let $\mu_{i}=h_{i}\widehat{G}_{i}$, for $1\leq i\leq e$.
\begin{lemma}\label{crt}
Assuming the same notations discussed above { and mod $(u^e-1)$}, we have $\sum_{i=1}^e\mu_{i}=1,$ and
\begin{equation*}
 \mu_{i}\mu_{j}=\begin{cases}
    \mu_{i}, & \text{if $i=j$}\\
    0, & \text{if $i\neq j$.}
  \end{cases}
\end{equation*}
Therefore, $\{\mu_1,\mu_2\dots,\mu_e\}$ is a set of orthogonal idempotent elements in $R_{e,q}$.
\end{lemma}
\begin{proof}
For $i\neq j$, we have $(u-\alpha_i)\mid h_{j}\widehat{G}_{j}$ and $(u-\alpha_j)\mid h_{i}\widehat{G}_{i}$. Thus, $\mu_i\mu_j=h_{i}\widehat{G}_{i}h_{j}\widehat{G}_{j}= 0 \pmod{u^e-1}$. To prove $\mu_i^2=\mu_i$, it is enough to show $(u-\alpha_s)\mid (\mu_i^2-\mu_i)$ for all $s=1,2,\dots,e$. For $i\neq s$, it follows from the construction. If $i=s$, then
\begin{align*}
    \mu_i^2-\mu_i&=\mu_i(\mu_i-1)\\
    &=h_{i}\widehat{G}_{i}(h_{i}\widehat{G}_{i}-1)\\
    &=z_iG_i(z_iG_i-1), ~(\text{since}~z_{i}G_{i}+h_{i}\widehat{G}_{i}=1).
\end{align*} Therefore, $(u-\alpha_i)\mid (\mu_i^2-\mu_i)$, and hence $\mu_i^2=\mu_i$.
{ Again,
\begin{align*}
 1-\sum_{i=1}^e\mu_i&=\prod_{i=1}^e(1-\mu_i)~(\text{since}~\mu_i^2=\mu_i,\mu_i\mu_j=0,i\neq j)\\
 &=\prod_{i=1}^e(1-h_i\widehat{G}_i)\\
 &=\prod_{i=1}^ez_iG_i\\
 &=0~\pmod{u^e-1}.
\end{align*}
Hence, $\sum_{i=1}^e\mu_i=1\pmod{u^e-1}$.}
\end{proof}
Taking into account Lemma \ref{crt} and the Chinese Remainder Theorem, we have the decomposition
\begin{align*}
R_{e,q}=\bigoplus_{i=1}^e\mu_{i}R_{e,q}\cong \bigoplus_{i=1}^e\mu_{i}\mathbb{F}_{q}.
\end{align*}
Therefore, any element $r\in R_{e,q}$ can be uniquely expressed as $r=\sum_{i=1}^e s_{i}\mu_{i}$, for $s_{i}\in \mathbb{F}_{q}, 1\leq i\leq e$. Then the units of $R_{e,q}$ can be classified by the result that $r$ is a unit in $R_{e,q}$ if and only if $s_i$ is a unit in $\mathbb{F}_q$, for $1\leq i\leq e$.

Let $GL_e(\mathbb{F}_q)$ be the set of all $e\times e$ invertible matrices over $\mathbb{F}_q$. Let $M\in GL_e(\mathbb{F}_q)$ such that $MM^T=\gamma I_e$, where $M^T$ is the transpose matrix of $M$, $I_e$ is the identity matrix of order $e$ and $\gamma\in \mathbb{F}_q^{*}=\mathbb{F}_q\setminus \{0\}$. { Since $q$ is odd, so for the existence of such matrix $M$, we restrict $\gamma\in \mathbb{F}_q^*$ such that $\gamma^e$ is a square in $\mathbb{F}_q^*$ (see [\cite{Hachenberger20}, Definition 7.6.4 and Proposition 7.6.6]).} Now, we define a map $\varphi: R_{e,q}\longrightarrow \mathbb{F}_{q}^{e}$ by
\begin{align}\label{map 1}
\varphi(r)=(s_{1},s_{2},\dots,s_{e})M=\boldsymbol{r}M,
\end{align}
where $r=\sum_{i=1}^e s_{i}\mu_{i}\in R_{e,q}$, for $s_{i}\in \mathbb{F}_{q}, 1\leq i\leq e$. Here, we use $\boldsymbol{r}$ for the vector $(s_{1},s_{2},\dots,s_{e})$. The map $\varphi$ can be extended to $R_{e,q}^{n}$ as follows:
\begin{align*}
\varphi: R_{e,q}^n\longrightarrow \mathbb{F}_{q}^{en}
\end{align*}
by
\begin{align}\label{map 1}
\nonumber \varphi(r_0,r_1,\dots,r_{n-1})=&[(s_{0,1},s_{0,2},\dots,s_{0,e})M,(s_{1,1},s_{1,2},\dots,s_{1,e})M,\dots ,(s_{n-1,1},s_{n-1,2},\dots,s_{n-1,e})M]\\
=&(\boldsymbol{r_0}M,\boldsymbol{r_1}M,\dots,\boldsymbol{r_{n-1}}M),
\end{align}
where $r_i=\sum_{j=1}^e s_{i,j}\mu_j$, for $0\leq i\leq n-1$.

The Hamming weight $w_{H}(c)$ is the number of non-zero components of the codeword $c=(c_{0},c_{1},\dots,c_{n-1})\in \mathcal{C}$ and the distance between two codewords is given by $d_{H}(c',c'')=w_{H}(c'-c'')$. The Hamming distance of $\mathcal{C}$ is defined as $d_{H}(\mathcal{C})=\min\{d_{H}(c',c'')\mid c'\neq c'',$  for all $c',c''\in \mathcal{C} \}$. The Gray weight of any element $r\in R_{e,q}$ is $w_{G}(r)=w_{H}(\varphi(r))$ and { the} Gray weight for $\overline{r}=(r_{0},r_{1},\dots,r_{n-1})\in R_{e,q}^{n}$ is $w_{G}(\overline{r})=\sum_{i=0}^{n-1}w_{G}(r_{i})$. The Gray distance between any two codewords $c',c''$ is defined as $d_{G}(c',c'')=w_{G}(c'-c'')$ and { the} Gray distance of $\mathcal{C}$ is $d_{G}(\mathcal{C})=\min\{d_{G}(c',c'')\mid c'\neq c'', c',c''\in \mathcal{C} \}$. Now, we present an example for $e=3,q=7$ to show the procedure of finding orthogonal { idempotent elements and the} corresponding Gray map.

\begin{example}

Let $e=3, t=2, q=3\cdot 2+1=7$ and $R_{3,7}=\mathbb{F}_7[u]/\langle u^3-1\rangle$. Then in $\mathbb{F}_7[u]$, we have
\begin{align*}
    u^3-1=(u+3)(u+5)(u+6).
\end{align*}
Let $G_1=u+3,G_2=u+5,G_3=u+6$. Then
\begin{align*}
    \widehat{G}_{1}=&\frac{u^3-1}{G_1}= (u+5)(u+6),\\
    \widehat{G}_{2}=&\frac{u^3-1}{G_2}=(u+3)(u+6),\\
    \widehat{G}_{3}=&\frac{u^3-1}{G_3}=(u+3)(u+5).
\end{align*}
Therefore,
\begin{align*}
    &h_1=6,z_1=(u+3),\mu_1=6(u+5)(u+6),\\
    &h_2=3,z_2=4(u+4),\mu_2=3(u+3)(u+6),\\
    &h_3=5,z_3=2(u+2),\mu_3=5(u+3)(u+5).
\end{align*}
Then $\{\mu_1,\mu_2,\mu_3\}$ is the set of orthogonal idempotent elements in $R_{3,7}=\mathbb{F}_7[u]/\langle u^3-1\rangle$, and hence by the Chinese Remainder Theorem, we have $R_{3,7}=\mu_{1}R_{3,7}\oplus\mu_{2}R_{3,7}\oplus\mu_{3}R_{3,7}\cong \mu_{1}\mathbb{F}_{7}\oplus\mu_{2}\mathbb{F}_{7}\oplus\mu_{3}\mathbb{F}_{7}$. Again, any element $r\in R_{3,7}$ can be written as
\begin{align*}
    r=a_0+ua_1+u^2a_2=&(a_0+4a_1+2a_2)\mu_1+(a_0+2a_1+4a_2)\mu_2+(a_0+a_1+a_2)\mu_3.
\end{align*} Thus, the map $\varphi:R_{3,7}\longrightarrow \mathbb{F}^3_7$ is given by \begin{align*}
    \varphi(r=a_0+ua_1+u^2a_2)=(a_0+4a_1+2a_2,a_0+2a_1+4a_2,a_0+a_1+a_2)M,
\end{align*}
where $a_i\in \mathbb{F}_7$, for $0\leq i\leq 2$ and $M\in GL_3(\mathbb{F}_7)$ with $MM^T=\gamma I_3$, $\gamma^3\in \mathbb{F}_7^*$ is a square.
\end{example}

{ In Table \ref{tabA}, we give
some examples of orthogonal idempotent set for different values of $e$
and $q$.} Also, we calculate the canonical components of $r\in R_{e,q}$ such that $r=\sum_{i=1}^e\mu_is_i$, which are useful in the subsequent examples.

{ Now, we remind
some well-known techniques from \cite{Gao15} as follows:} Let $A_{i}$ be a linear code for $1\leq i\leq e$, we denote $\bigoplus_{i=1}^eA_i=\{a_{1}+a_{2}+\dots+a_{e}\mid a_{i}\in A_{i}~\forall~i \}$. Let $\mathcal{C}$ be a linear code of length $n$ over $R_{e,q}$. For $1\leq i\leq e$, define
$\mathcal{C}_{i}:=\{s_{i}\in \mathbb{F}_{q}^{n}\mid \exists ~s_{1},s_{2},\dots,s_{i-1},s_{i+1},\dots,s_e$ such that $\sum_{i=1}^{e}s_{i}\mu_{i}\in \mathcal{C} \}$. Then $\mathcal{C}_{i}$ is { a} linear code of length $n$ over $\mathbb{F}_{q}$ for $1\leq i\leq e$. Also, the code $\mathcal{C}$ can be uniquely expressed as $\mathcal{C}=\bigoplus_{i=1}^e\mu_{i}\mathcal{C}_{i}$ and $\mid \mathcal{C}\mid=\prod_{i=1}^e\mid \mathcal{C}_{i}\mid$. Moreover, a generator matrix for the code $\varphi(\mathcal{C})$ can be obtained as $M=\begin{pmatrix}
 \varphi(\mu_{1}M_{1}) \\
  \varphi(\mu_{2}M_{2}) \\
  \vdots\\
  \varphi(\mu_{e}M_{e})
\end{pmatrix}$,
where $M_{i}$ is a generator matrix of $\mathcal{C}_{i}$, for $1\leq i\leq e$.
\begin{lemma}
The map $\varphi$ defined in equation (\ref{map 1}) is linear and isometric from ($R_{e,q}^{n}$, Gray distance) to ($\mathbb{F}_{q}^{en}$, Hamming distance).
\end{lemma}
\begin{proof}
{ It follows from the definition of the Gray weight.}
\end{proof}

\begin{theorem}\label{thm gray image}
Let $\mathcal{C}$ be a linear code of length $n$  over $R_{e,q}$, $\mid \mathcal{C}\mid=q^{ek}$ and Gray distance $d_G$.
\begin{enumerate}
    \item Then $\varphi(\mathcal{C})$ is an $[en,k,d_H]$ linear code over $\mathbb{F}_{q}$ where $d_{H}=d_{G}$.
    \item $\varphi(\mathcal{C}^{\perp})=(\varphi(\mathcal{C}))^{\perp}$.
    \item $\varphi(\mathcal{C})$ is a self-orthogonal (Euclidean) linear code of length $en$ over $\mathbb{F}_q$ if $\mathcal{C}$ is a self-orthogonal (Euclidean) linear code over $R_{e,q}$.
    \item $\varphi(\mathcal{C})$ is a self-dual code if and only if $\mathcal{C}$ is a self-dual code.
\end{enumerate}
\end{theorem}

\begin{proof}
\begin{enumerate}
    \item Since $\varphi$ is a linear map, $\varphi(\mathcal{C})$ is a linear code of length $en$ over $\mathbb{F}_{q}$. Also, $\varphi$ being isometric, $\varphi(\mathcal{C})$ is an $[en,k,d_H]$ linear code over $\mathbb{F}_{q}$ with $d_{H}=d_{G}$.
    \item

    Let $c=(c_0,c_1,\dots,c_{n-1})\in \mathcal{C}$ { and} $d=(d_0,d_1,\dots,d_{n-1})\in \mathcal{C}^{\perp}$, where $c_j=\sum_{i=1}^e\mu_it_j^i,~d_j=\sum_{i=1}^e\mu_im_j^i$, $t^i_j,m^i_j\in \mathbb{F}_q$ for all $i,j$. Now, $c\cdot d=\sum_{j=0}^{n-1}c_jd_j=0$ gives, $\sum_{j=0}^{n-1}(t^1_jm^1_j+t^2_jm^2_j+\dots+t^e_jm^e_j)=0$. Again,
    \begin{align*}
        &\varphi(c)=[(t_0^1,t_0^2,\dots,t_0^e)M,\dots, (t_{n-1}^1,t_{n-1}^2,\dots,t_{n-1}^e)M]=(\alpha_0M,\dots,\alpha_{n-1}M)\\
        &\text{and}\\
        &\varphi(d)=[(m_0^1,m_0^2,\dots,m_0^e)M,\dots, (m_{n-1}^1,m_{n-1}^2,\dots,m_{n-1}^e)M]=(\beta_0M,\dots,\beta_{n-1}M),
    \end{align*}{}
where $\alpha_j=(t_j^1,t_j^2,\dots,t_j^e),$ $\beta_j=(m_j^1,m_j^2,\dots,m_j^e)$ for $0\leq j\leq n-1$ and $MM^T=\gamma I_e$. Now,
\begin{align*}
    \varphi(c)\cdot \varphi(d)=\varphi(c) \varphi(d)^T=\sum_{j=0}^{n-1}\alpha_jMM^T\beta_j^T=\gamma\sum_{j=0}^{n-1}\alpha_j\beta_j^T=\gamma\sum_{j=0}^{n-1}(t^1_jm^1_j+t^2_jm^2_j+\dots+t^e_jm^e_j)=0.
\end{align*}
Since $c\in \mathcal{C}$ and $d\in \mathcal{C}^{\perp}$ { are} arbitrary, so $\varphi(\mathcal{C}^{\perp})\subseteq (\varphi(\mathcal{C}))^{\perp}$. On the other side, as $\varphi$ is a bijective linear map, so $\mid \varphi(\mathcal{C}^{\perp})\mid= \mid (\varphi(\mathcal{C}))^{\perp}\mid$. Therefore, $\varphi(\mathcal{C}^{\perp})=(\varphi(\mathcal{C}))^{\perp}$

\item Let $\mathcal{C}$ be self-orthogonal, that is, $\mathcal{C}\subseteq \mathcal{C}^{\perp}$. Now, $\varphi(\mathcal{C})\subseteq \varphi(\mathcal{C}^{\perp})=(\varphi(\mathcal{C}))^{\perp}$. Hence $\varphi(\mathcal{C})$ is a self-orthogonal linear code of length $en$ over $\mathbb{F}_q$.
\item Let $\mathcal{C}$ be self-dual, that is $\mathcal{C}=\mathcal{C}^{\perp}$. Then $\varphi(\mathcal{C})=\varphi(\mathcal{C}^{\perp})=(\varphi(\mathcal{C}))^{\perp}$. Hence $\varphi(\mathcal{C})$ is a self-dual linear code over $\mathbb{F}_q$. Conversely, let $\varphi(\mathcal{C})$ be a self-dual code. Then $\varphi(\mathcal{C})=(\varphi(\mathcal{C}))^{\perp}=\varphi(\mathcal{C}^{\perp})$. Since $\varphi$ is bijection, $\mathcal{C}=\mathcal{C}^{\perp}$. Hence $\mathcal{C}$ is a self-dual code of length $n$ over $R_{e,q}$.

\end{enumerate}
\end{proof}
Now, to obtain the necessary and sufficient condition for { the} existence of self-dual codes we calculate dual codes in { the} next result, analogous to Theorem 5 of \cite{Zheng17}.
\begin{theorem}\label{th:2}
Let $\mathcal{C}=\bigoplus_{i=1}^e\mu_{i}\mathcal{C}_{i}$ be a linear code of length $n$ over $R_{e,q}$. Then
\begin{enumerate}
    \item $\mathcal{C}^{\perp}=\bigoplus_{i=1}^e\mu_{i}\mathcal{C}_{i}^{\perp}$, and
    \item  $\mathcal{C}$ is self-dual (Euclidean) if and only if $\mathcal{C}_{i}$ is self-dual (Euclidean), for $1\leq i\leq e$ .
\end{enumerate}
\end{theorem}

\begin{proof}
\begin{enumerate}
\item Let $\mathcal{C}$ be a linear code of length $n$ over $R_{e,q}$. Also, let
\begin{align*}
    \mathcal{S}_{i}:=\{{ s_{i}}\in \mathbb{F}_{q}^{n}\mid \exists ~s_{1},s_{2},\cdots,s_{i-1},s_{i+1},\cdots,s_e ~\text{such that}~ \sum_{i=1}^{e}s_{i}\mu_{i}\in \mathcal{C}^{\perp} \}, ~\text{for}~ 1\leq i\leq e.
\end{align*}
Thus, $\mathcal{C}^{\perp}$ can be uniquely written as $\mathcal{C}^{\perp}=\bigoplus_{i=1}^e\mu_{i}\mathcal{S}_{i}$. Clearly, $\mathcal{S}_{1}\subseteq \mathcal{C}_{1}^{\perp}$. On the other side, let $r\in \mathcal{C}_{1}^{\perp}$. Then $r\cdot a_{1}=0$ for all $a_{1}\in \mathcal{C}_{1}$. Let $z=\sum_{i=1}^{e}\mu_{i}a_{i}\in \mathcal{C}$. Now, $\mu_{1}r\cdot z=\mu_{1}a_{1}\cdot r=0$, which implies $\mu_{1}r\in \mathcal{C}^{\perp}$. From the unique expression of $\mathcal{C}^{\perp}$, we have $r\in \mathcal{S}_{1}$. Therefore, $\mathcal{C}_{1}^{\perp}\subseteq \mathcal{S}_{1}$. Hence, $\mathcal{S}_{1}= \mathcal{C}_{1}^{\perp}$. Similarly, $\mathcal{C}_{i}^{\perp}=\mathcal{S}_{i}$ for $2\leq i\leq e$. Consequently, $\mathcal{C}^{\perp}=\bigoplus_{i=1}^e\mu_{i}\mathcal{C}_{i}^{\perp}$.\\
\item Further, $\mathcal{C}$ is self-dual if and only if $ \mathcal{C}^{\perp}=\mathcal{C}$ if and only if $ \bigoplus_{i=1}^e\mu_{i}\mathcal{C}_{i}^{\perp}=\bigoplus_{i=1}^e\mu_{i}\mathcal{C}_{i}$. Therefore, $\mathcal{C}_{i}^{\perp}=\mathcal{C}_{i}$ for $1\leq i\leq e$.
\end{enumerate}
\end{proof}

\section{Cyclic codes over $R_{e,q}$}\label{sec 3}
In this section, we determine the structure of cyclic codes of length $n$ over $R_{e,q}$ and use these results to obtain LCD codes in the subsequent section. Note that in \cite{MNR} the authors considered this type of codes as codes over an affine algebras with a finite commutative chain coefficient ring where the affine algebra is $R_{e,q}$. However, we will construct them from scratch based on the results before so we get an explicit description of the duality and LCD condition based on Theorem~\ref{th:2} above.
\begin{df}
A linear code $\mathcal{C}$ of length $n$ over $R_{e,q}$ is said to be a cyclic code if for any $c=(c_{0},c_{1},\dots,c_{n-1})\in \mathcal{C}$, we have $\sigma(c):=(c_{n-1},c_{0},\dots,c_{n-2})\in \mathcal{C}$. Here, $\sigma$ is called the cyclic shift operator.
\end{df}
We identify each codeword $c=(c_{0},c_{1},\dots,c_{n-1})\in \mathcal{C}$ by a polynomial $c(x)\in R_{e,q}[x]/\langle x^{n}-1\rangle$ under the correspondence $c=(c_{0},c_{1},\dots,c_{n-1})\longmapsto c(x)=c_{0}+c_{1}x+\dots+c_{n-1}x^{n-1} \pmod{x^{n}-1}$. Under this identification, we have the following result.

\begin{theorem}
Let $\mathcal{C}$ be a linear code of length $n$ over $R_{e,q}$. Then $\mathcal{C}$ is a cyclic code if and only if $\mathcal{C}$ is an ideal of the ring $R_{e,q}[x]/\langle x^{n}-1\rangle$.
\end{theorem}
\begin{proof}
Let $\mathcal{C}$ be a cyclic code of length $n$ over $R_{e,q}$. Let $a(x)\in R_{e,q}[x]$ and $c(x)=c_0+c_1x+\dots+ c_{n-1}x^{n-1}\in \mathcal{C}$. Then $xc(x)=c_0x+c_1x^2+\dots+ c_{n-2}x^{n-1}+c_{n-1}x^n=c_{n-1}+c_0x+c_1x^2+\dots +c_{n-2}x^{n-1} { \pmod{x^n-1}}=\sigma(c)\in \mathcal{C}$. In similar manner, for any positive integer $i\geq 2$, we have $x^ic(x)\in \mathcal{C}$. Since $\mathcal{C}$ is an $R_{e,q}$-submodule, $a(x)c(x)\in \mathcal{C}$. Hence, $\mathcal{C}$ is an ideal of $R_{e,q}[x]/\langle x^{n}-1\rangle$.

Conversely, let $\mathcal{C}$ be an ideal of $R_{e,q}[x]/\langle x^{n}-1\rangle$. Let $c=(c_0,c_1,\dots,c_{n-1})\in \mathcal{C}$. Now, $\sigma(c)=c_{n-1}+c_0x+c_1x^2+\dots+ c_{n-2}x^{n-1}=xc(x)\in \mathcal{C}$. Hence $\mathcal{C}$ is a cyclic code of length $n$ over $R_{e,q}$.
\end{proof}

\begin{theorem}\cite{Hill} \label{thm gen cyclic field}
Let $\mathcal{C}$ be a cyclic code of length $n$ over $\mathbb{F}_q$. Then there exists a unique polynomial $g(x)\in \mathbb{F}_q[x]$ such that $\mathcal{C}=\langle g(x)\rangle$ and $\mathcal{C}^{\perp}=\langle h^*(x)\rangle$ where $x^n-1=g(x)h(x)$ and $h^*(x)$ is the reciprocal polynomial of $h(x)$.
\end{theorem}

\begin{theorem}\label{thm dec 1}
A linear code $\mathcal{C}=\bigoplus_{i=1}^e\mu_{i}\mathcal{C}_{i}$ of length $n$ over $R_{e,q}$ is cyclic if and only if $\mathcal{C}_{i}$ is { a} cyclic code of length $n$ over $\mathbb{F}_{q}$, for $1\leq i\leq e$.
\end{theorem}
\begin{proof}
Let $\mathcal{C}=\bigoplus_{i=1}^e\mu_{i}\mathcal{C}_{i}$ be a cyclic code of length $n$ over $R_{e,q}$. Let $a^i=(a_0^i,a_1^i,\dots,a_{n-1}^i)\in \mathcal{C}_i$ for $1\leq i\leq e$. Take $r_j=\sum_{i=1}^{e}\mu_ia^i_j$ for $0\leq j\leq n-1$. Then $r=(r_0,r_1,\dots,r_{n-1})\in \mathcal{C}$ and hence $\sigma(r)=(r_{n-1},r_0,\dots,r_{n-2})\in \mathcal{C}$. Again $\sigma(r)=\sum_{i=1}^e\mu_i\sigma(a^i)\in \mathcal{C}=\bigoplus_{i=1}^e\mu_{i}\mathcal{C}_{i}$. Therefore, $\sigma(a^i)\in \mathcal{C}_i$ for $1\leq i\leq e$. Hence $\mathcal{C}_i$ is { a} cyclic code over $\mathbb{F}_q$ for $1\leq i\leq e$.

Conversely, let $\mathcal{C}_i$ be a cyclic code over $\mathbb{F}_q$ for $1\leq i\leq e$. Let $r=(r_0,r_1,\dots,r_{n-1})\in \mathcal{C}$ where $r_j=\sum_{i=1}^{e}\mu_ia^i_j$ for $0\leq j\leq n-1$. Then $a^i=(a_0^i,a_1^i,\dots,a_{n-1}^i)\in \mathcal{C}_i$ for $1\leq i\leq e$. Now, $\sigma(r)=\sum_{i=1}^e\mu_i\sigma(a^i)\in \bigoplus_{i=1}^e\mu_{i}\mathcal{C}_{i}=\mathcal{C}$. Hence $\mathcal{C}$ is a cyclic code of length $n$ over $R_{e,q}$.
\end{proof}

By the help of Theorem \ref{thm gen cyclic field} and Theorem \ref{thm dec 1} we prepare Theorem \ref{thm dec gen 1} to obtain the generator polynomial for cyclic codes.

\begin{theorem}\label{thm dec gen 1}
Let $\mathcal{C}=\bigoplus_{i=1}^e\mu_{i}\mathcal{C}_{i}$ be a cyclic code of length $n$ over $R_{e,q}$. Then there exists a unique polynomial $g(x)\in R_{e,q}[x]$ such that $\mathcal{C}=\langle g(x)\rangle$ and $x^n-1=g(x)h(x)$ where $g(x)=\sum_{i=1}^{e}\mu_ig_i(x)$ and $x^{n}-1=g_i(x)h_i(x)$, for $1\leq i\leq e$.
 \end{theorem}
 \begin{proof}
Let $\mathcal{C}=\bigoplus_{i=1}^e\mu_{i}\mathcal{C}_{i}$ be a  cyclic code of length $n$ over $R_{e,q}$.  Then by Theorem \ref{thm dec 1}, $\mathcal{C}_i$ { is a cyclic code} over $\mathbb{F}_q$ for $1\leq i\leq e$. By Theorem \ref{thm gen cyclic field}, we have $\mathcal{C}_i=\langle g_i(x)\rangle$ where $x^n-1=g_i(x)h_i(x)$ for $1\leq i\leq e$. Let $g(x)=\sum_{i=1}^e\mu_ig_i(x)$. Then $\langle g(x)\rangle\subseteq \mathcal{C}$. On the other hand, $\mu_ig_i(x)=\mu_ig(x)\in \langle g(x)\rangle$ for $1\leq i\leq e$. Thus $\mathcal{C}\subseteq \langle g(x)\rangle$. Therefore, $\mathcal{C}=\langle g(x)\rangle$. Since $g_i(x)$ is unique for $1\leq i\leq e$, then $g(x)$ is unique.\\
Further, we have $g(x)\sum_{i=1}^e\mu_ih_i(x)=\sum_{i=1}^e\mu_ig_i(x)h_i(x)=x^n-1\sum_{i=1}^e\mu_i=x^n-1$. Hence $x^n-1=g(x)h(x)$ where $h(x)=\sum_{i=1}^e\mu_ih_i(x)$.
 \end{proof}

\begin{corollary}
Every cyclic code of length $n$ over $R_{e,q}$ is principally generated, that is, $R_{e,q}[x]/\langle x^n-1\rangle$ is a principal ideal ring.
\end{corollary}

\begin{corollary} Let $\mathcal{C}=\bigoplus_{i=1}^e\mu_{i}\mathcal{C}_{i}$ be a  cyclic code of length $n$ over $R_{e,q}$ and $\mathcal{C}_{i}=\langle g_{i}(x)\rangle$ such that $x^{n}-1=g_{i}(x)h_{i}(x)$ for $1\leq i\leq e$. Then
\begin{enumerate}
\item $\mathcal{C}^{\perp}=\bigoplus_{i=1}^e\mu_{i}\mathcal{C}_{i}^{\perp}$ is a  cyclic code over $R_{e,q}$,
\item $\mathcal{C}^{\perp}=\langle \sum_{i=1}^e\mu_{i}h_{i}^{*}(x)\rangle$, where $h_{i}^{*}(x)$ is the reciprocal polynomial of $h_{i}(x)$, i.e., $h_{i}^{*}(x)=x^{\deg(h_{i}(x))}h_{i}(1/x)$ for $1\leq i\leq e$,
\item $\mid \mathcal{C}^{\perp}\mid=q^{\sum_{i=1}^{e}\deg(g_{i}(x))}$.
\end{enumerate}
\end{corollary}
\begin{proof}
\begin{enumerate}
\item Let $\mathcal{C}=\bigoplus_{i=1}^e\mu_{i}\mathcal{C}_{i}$ be a  cyclic code of length $n$ over $R_{e,q}$. Then by Theorem \ref{thm dec 1}, $\mathcal{C}_i$ is { a} cyclic code over $\mathbb{F}_q$ for $1\leq i\leq e$. Then $\mathcal{C}_i^{\perp}$ is { a} cyclic code over $\mathbb{F}_q$ for $1\leq i\leq e$. By Theorem \ref{thm dec 1}, $\mathcal{C}^{\perp}=\bigoplus_{i=1}^e\mu_{i}\mathcal{C}_{i}^{\perp}$ is a  cyclic code over $R_{e,q}$.

    \item Let $\mathcal{C}=\bigoplus_{i=1}^e\mu_{i}\mathcal{C}_{i}$ be a  cyclic code of length $n$ over $R_{e,q}$. Then by Theorem \ref{thm gen cyclic field}, $\mathcal{C}_i^{\perp}=\langle h^*_i(x)\rangle$ where $x^n-1=g_i(x)h_i(x)$ for $1\leq i\leq e$. Therefore, similar to the proof of Theorem \ref{thm dec gen 1}, we have $\mathcal{C}^{\perp}=\langle \sum_{i=1}^e\mu_ih^*_i(x)\rangle$, where $h_{i}^{*}(x)$ is the reciprocal polynomial of $h_{i}(x)$.
    \item Since $\mid \mathcal{C}_i^{\perp}\mid =q^{\deg(g_i(x))}$, so $\mid \mathcal{C}^{\perp}\mid =\prod_{i=1}^e\mid \mathcal{C}_i^{\perp}\mid =q^{\sum_{i=1}^e\deg(g_i(x))}$.
\end{enumerate}{}
\end{proof}

Now, we present some examples of cyclic codes over $R_{e,q}$ and their $\mathbb{F}_q$-images which include MDS and optimal linear codes. The calculation involved in these examples is carried out by the Magma computation system \cite{Magma}.
\begin{example}
Let $e=2,q=5,n=10$ and $R_{2,5}=\mathbb{F}_5[u]/\langle u^2-1\rangle$. Then $\mu_1=\frac{1+u}{2}$ and $\mu_2=\frac{1-u}{2}$. Let  \[
   M=
  \left[ {\begin{array}{cccc}
   3&2 \\
   2&2
  \end{array} } \right]\in GL_2(\mathbb{F}_5).
\]
Then $M$ satisfies $MM^T=3I_2$ and { the} Gray map $\varphi: R_{2,5}\longrightarrow \mathbb{F}_5^2$ is defined by
\begin{align*}
    \varphi(a_0+ua_1)=(a_0+a_1,a_0-a_1)\left[ {\begin{array}{cccc}
   3&2 \\
   2&2
  \end{array} } \right].
\end{align*}
Again
\begin{align*}
    x^{10}-1= (x + 1)^5(x + 4)^5\in \mathbb{F}_5[x].
\end{align*}{}
Let $\mathcal{C}=\langle \mu_1g_1(x)+\mu_2g_2(x)\rangle$ be a cyclic code of length $10$ over $R_{2,5}$ where $g_1(x)=x + 4$ and $g_2(x)=(x+1)(x+4)^3=x^4 + 3x^3 + 2x + 4$. Therefore, $\varphi(\mathcal{C})$ is a $[20,15,4]$ linear code over $\mathbb{F}_5$ and as per the database \cite{Grassltable}, it is an optimal linear code.
\end{example}

\begin{example}
Let $e=2,q=11,n=5$ and $R_{2,11}=\mathbb{F}_{11}[u]/\langle u^2-1\rangle$. Now, in $\mathbb{F}_{11}[x]$, we have
\begin{align*}
    x^{5}-1=(x + 2)(x + 6)(x + 7)(x + 8)(x + 10).
\end{align*}
Let $g_1(x)=(x+8)(x+10)=x^2 + 7x + 3$ and $g_2(x)=(x + 6)(x+7)=x^2 + 2x + 9$. Then $\mathcal{C}=\langle \mu_1g_1(x)+\mu_2g_2(x)\rangle$ { is} a cyclic code of length $5$ over $R_{2,11}$. Again, let \[
   M=
  \left[ {\begin{array}{cccc}
   9&2 \\
   2&2
  \end{array} } \right]\in GL_2(\mathbb{F}_{11}),
\]
satisfying $MM^T=8I_2$. Therefore, the Gray image $\varphi(\mathcal{C})$ { has parameters} $[10,6,5]$, which is an MDS code satisfying $n-k+1-d=0$.
\end{example}

\begin{example}
Let $\mathcal{C}=\langle \mu_1g_1(x)+\mu_2g_2(x)\rangle$ be a cyclic code of length $n$ over $R_{2,q}=\mathbb{F}_q[u]/\langle u^2-1\rangle$, where $g_i(x)\mid x^n-1$ for $i=1,2$. Let  \[
   M=
  \left[ {\begin{array}{cccc}
   -2&2 \\
   2&2
  \end{array} } \right]\in GL_2(\mathbb{F}_{q}),
\]
satisfying $MM^T=8I_2$. Now, in Table \ref{tab1} and Table \ref{tab2}, the generator polynomials $g_i(x)~(i=1,2)$ are presented in the third and fourth column, respectively and the Gray images $\varphi(\mathcal{C})$ are included in the fifth column. As per the database \cite{Grassltable}, the codes obtained in the fifth column of Table \ref{tab1} are optimal and { the} best-known linear codes (BKLC). On the other hand, the codes obtained in the fifth column of Table \ref{tab2} attaining the singleton bound $n-k+1-d=0$ are MDS.
\end{example}

\section{LCD codes over $R_{e,q}$}
In { the} previous section we characterized cyclic codes of length $n$ over $R_{e,q}$. Here, using their structure we construct LCD codes { of length $n$ (Corollary \ref{cor lcd} and \ref{cor rev})}. Moreover, we discuss about the existence of free and non-free LCD codes. Recall that a linear code $\mathcal{C}$ of length $n$ is said to be a \textit{linear complementary dual} (shortly, LCD) code if $\mathcal{C}\cap \mathcal{C}^{\perp}=\{0\}$.

\begin{df}
A cyclic code $\mathcal{C}$ is said to be a reversible cyclic code if for any $(c_0,c_1,\dots,c_{n-1})\in \mathcal{C}$, its reverse translated codeword $(c_{n-1},c_{n-2},\dots,c_1,c_0)$ also is in $\mathcal{C}$.
\end{df}
Recall that a cyclic code of length $n$ over a field $\mathbb{F}_q$ is a reversible cyclic code if it is generated by a self-reciprocal factor of $x^n-1$. Also, LCD codes over  $\mathbb{F}_q$ are characterized by Yang and Massey \cite{Yang94}. Using these results here we determine the structure of LCD codes over $R_{e,q}$. Now, we recall the following results.

\begin{lemma}\cite{Yang94} \label{lem lcd field}
Let $\mathcal{C}=\langle g(x)\rangle$ be a cyclic code of length $n$ over $\mathbb{F}_q$.
\begin{enumerate}
\item Let $\gcd(n,q)\neq 1$. Then $\mathcal{C}$ is an LCD code if and only if $g(x)$ is self-reciprocal and all monic irreducible factors of $g(x)$ have the same multiplicity in $g(x)$ and $x^n-1$.

    \item Let $\gcd(n,q)=1$. Then { $\mathcal{C}$ is an LCD code if and only if it is a reversible code}.

\end{enumerate}
\end{lemma}

\begin{theorem}\label{thm lcd ring}
Let $\mathcal{C}=\bigoplus_{i=1}^e\mu_{i}\mathcal{C}_{i}$ be a cyclic code of length $n$ over $R_{e,q}$. Then $\mathcal{C}$ is an LCD code if and only if $\mathcal{C}_i$ is an { LCD cyclic code of length} $n$ over $\mathbb{F}_{q}$ for $1\leq i\leq e$.
\end{theorem}
\begin{proof}
Note that $\mathcal{C}\cap \mathcal{C}^{\perp}=\{0\}$ if and only if $\mathcal{C}_i\cap \mathcal{C}_i^{\perp}=\{0\}$. Hence the result follows.
\end{proof}

\begin{corollary}\label{cor lcd}
Let $\gcd(n,q)\neq 1$. Also, with the notations of Theorem \ref{thm dec gen 1}, let $\mathcal{C}=\langle \sum_{i=1}^{e}\mu_ig_i(x)\rangle$ be a cyclic code of length $n$ over $R_{e,q}$, where  $g_{i}(x)\in \mathbb{F}_{q}[x]$ and $g_{i}(x)\mid (x^{n}-1)$ for $1\leq i\leq e$. Then $\mathcal{C}$ is an LCD code if and only if $g_{i}(x)$ is self-reciprocal and each monic irreducible factor of $g_{i}(x)$ has the same multiplicity in $g_{i}(x)$ and in $(x^{n}-1)$, for $1\leq i\leq e$.
\end{corollary}
\begin{proof}
It follows from part 1 of Lemma \ref{lem lcd field} and Theorem \ref{thm lcd ring}.
\end{proof}

\begin{theorem}\label{thm reversible}
Let $\gcd(n,q)=1$ and $\mathcal{C}=\bigoplus_{i=1}^e\mu_{i}\mathcal{C}_{i}$ be a cyclic code of length $n$ over $R_{e,q}$. Then $\mathcal{C}$ is an LCD code if and only if $\mathcal{C}_{i}$ is a reversible code of length $n$ over $\mathbb{F}_{q}$ for $1\leq i\leq e$.
\end{theorem}
\begin{proof}
It follows from the combination of part 2 of Lemma \ref{lem lcd field} and Theorem \ref{thm lcd ring}.
\end{proof}
\begin{corollary}\label{cor rev}
Let $\gcd(n,q)=1$ and $\mathcal{C}=\bigoplus_{i=1}^e\mu_{i}\mathcal{C}_{i}$ be a cyclic code of length $n$ over $R_{e,q}$ where $\mathcal{C}_{i}=\langle g_{i}(x)\rangle$ is a cyclic code of length $n$ over $\mathbb{F}_{q}$ for $1\leq i\leq e$. Then $\mathcal{C}$ is an LCD code if and only if $g_{i}(x)$ is self-reciprocal polynomial in $\mathbb{F}_{q}$, for $1\leq i\leq e$.
\end{corollary}

\begin{lemma}\label{lemma lcd gray}
Let $\mathcal{C}$ be a linear code of length $n$ over $R_{e,q}$. Then $\varphi(\mathcal{C}\cap \mathcal{C}^{\perp})=\varphi(\mathcal{C})\cap \varphi(\mathcal{C})^{\perp}$, where $\varphi$ is the Gray map defined in equation (\ref{map 1}).
\end{lemma}
\begin{proof}
Let $a\in \varphi(\mathcal{C})\cap \varphi(\mathcal{C}^{\perp})$. Then there exist $c\in \mathcal{C}$ and $c'\in \mathcal{C}^{\perp}$ such that $\varphi(c)=\varphi(c')=a$. Now, $\varphi$ being injective, we have $c=c'\in \mathcal{C}\cap \mathcal{C}^{\perp}$. Therefore, $a=\varphi(c)\in \varphi(\mathcal{C}\cap \mathcal{C}^{\perp})$. Hence, $\varphi(\mathcal{C})\cap \varphi(\mathcal{C}^{\perp})\subseteq \varphi(\mathcal{C}\cap \mathcal{C}^{\perp})$. \\
On the other side, let $a\in \varphi(\mathcal{C}\cap \mathcal{C}^{\perp})$. Then there exists $c\in \mathcal{C}\cap \mathcal{C}^{\perp}$ such that $\varphi(c)=a$. Again, $c\in \mathcal{C}\cap \mathcal{C}^{\perp}$ implies that $c\in \mathcal{C}$ and $c\in \mathcal{C}^{\perp}$. Therefore, $\varphi(c)\in \varphi(\mathcal{C})$ and $\varphi(c)\in \varphi(\mathcal{C}^{\perp})$. Hence, $a=\varphi(c)\in \varphi(\mathcal{C})\cap \varphi(\mathcal{C}^{\perp})$. This shows that $\varphi(\mathcal{C}\cap \mathcal{C}^{\perp})\subseteq \varphi(\mathcal{C})\cap \varphi(\mathcal{C}^{\perp})$. Hence, combining both sides we have  $\varphi(\mathcal{C})\cap \varphi(\mathcal{C}^{\perp})= \varphi(\mathcal{C}\cap \mathcal{C}^{\perp})$. Also, by Theorem  \ref{thm gray image}, $\varphi(\mathcal{C}^{\perp})=\varphi(\mathcal{C})^{\perp}$. Thus, $\varphi(\mathcal{C})\cap \varphi(\mathcal{C})^{\perp}= \varphi(\mathcal{C}\cap \mathcal{C}^{\perp})$.
\end{proof}

\begin{theorem}\label{thm gr lcd}
Let $\mathcal{C}$ be a linear code of length $n$ over $R_{e,q}$. Then $\mathcal{C}$ is an LCD code if and only if its Gray image $\varphi(\mathcal{C})$ is an LCD code of length $en$ over $\mathbb{F}_{q}$.
\end{theorem}
\begin{proof}
Let $\mathcal{C}$ be an LCD code of length $n$ over $R_{e,q}$. Then $\mathcal{C}\cap \mathcal{C}^{\perp}=\{0\}$. Now, by Lemma \ref{lemma lcd gray}, $\varphi(\mathcal{C})\cap \varphi(\mathcal{C})^{\perp}=\varphi(\mathcal{C}\cap \mathcal{C}^{\perp})=\varphi(\{0\})=\{0\}$. This implies that $\varphi(\mathcal{C})$ is an LCD code of length $en$ over $\mathbb{F}_{q}$.

Conversely, let $\varphi(\mathcal{C})$ be an LCD code of length $en$ over $\mathbb{F}_{q}$. Then $\varphi(\mathcal{C})\cap \varphi(\mathcal{C})^{\perp}=\{0\}$. Also, by Lemma \ref{lemma lcd gray}, we have $\varphi(\mathcal{C}\cap \mathcal{C}^{\perp})=\varphi(\mathcal{C})\cap \varphi(\mathcal{C})^{\perp}=\{0\}$. Since $\varphi$ is injective, we have $\mathcal{C}\cap \mathcal{C}^{\perp}=\{0\}$. Hence, $\mathcal{C}$ is an LCD code of length $n$ over $\mathcal{R}$.
\end{proof}
\subsection{Existence of non-free LCD codes over $R_{e,q}$}
{ In \cite{Bhowmick20},} the authors proved that there does not exist any non-free LCD codes over finite commutative local Frobenius ring, or in other words, every LCD code over commutative local { Frobenius ring is free. Also, the authors gave the conditions in case} of a non-local ring for a code to be LCD. Here we show that both free and non-free LCD cyclic codes exist over the commutative semi-local Frobenius ring $R_{e,q}$ in terms of the degree of the polynomials defining them. To do so we use the following result.
\begin{lemma}\label{idem element}
The element $\sum_{i\in T}\mu_i$ is a zero divisor in $R_{e,q}$, for any non-trivial subset $T$ of $\{1,2,\dots,e\}$.
\end{lemma}
\begin{proof}
Since $\mu_i(1-\mu_i)=0$, for $1\leq i\leq e$, then $\sum_{i\in T}\mu_i\prod_{i\in T}(1-\mu_i)=0$, for any non-trivial subset $T$ of $\{1,2,\dots,e\}$. Also, as $\mu_i\mu_j=0$, for $i\neq j$, then $\prod_{i\in T}(1-\mu_i)=1-(\sum_{i\in T}\mu_i)\neq 0$. Hence the result.
\end{proof}
As we { discussed} in the previous section a cyclic code of length $n$ is given by $\mathcal{C}=\langle g(x)\rangle$ where $g(x)=\sum_{i=1}^e\mu_ig_i(x)$. Since $g_i(x)$ is { a} monic polynomial in $\mathbb{F}_q[x]$, $g(x)$ is monic if the degrees of all polynomials $g_i(x)$, for $1\leq i\leq e$ are same. Also, if the degrees of all polynomials $g_i(x)$ are not same, then by Lemma \ref{idem element}, $g(x)$ is a polynomial having a zero divisor leading coefficient. Therefore, $g(x)$ is { a} non-monic polynomial. Now, for the case when $g(x)$ is monic, by using { the} division algorithm on $R_{e,q}[x]$ and { a} similar argument of Theorem 4 in \cite{Sharma17}, we conclude that $\mathcal{C}$ is free. On the other hand, when $g(x)$ is non-monic, if possible, { let us assume that $\mathcal{C}$ is free.} Then by similar argument of Theorem 5 in \cite{Sharma17}, we have a monic polynomial $g'(x)\in R_{e,q}[x]$ such that $\mathcal{C}=\langle g'(x)\rangle$ and $g'(x)\mid x^n-1$. Since $g(x)$ is non-monic, this contradicts the fact that $g(x)$ is the unique polynomial in $R_{e,q}[x]$ such that $\mathcal{C}=\langle g(x)\rangle$ and $g(x)\mid x^n-1$. Therefore, under above discussion, we have the following results.
\begin{theorem}\label{thm free}
Let $\mathcal{C}=\langle g(x)\rangle$ be a cyclic code of length $n$ over $R_{e,q}$ where $g(x)=\sum_{i=1}^e\mu_ig_i(x)$.
\begin{enumerate}
    \item $\mathcal{C}$ is free if the degrees of $g_1(x),g_2(x),\dots,g_e(x)$ are same.
    \item $\mathcal{C}$ is non-free if the degrees of $g_1(x),g_2(x),\dots,g_e(x)$ are not same.
\end{enumerate}
\end{theorem}
Unlike commutative local Frobenius ring, both free and non-free LCD codes exist over $R_{e,q}$ and { can be constructed} by following Corollary \ref{cor lcd}, Corollary \ref{cor rev} and Theorem \ref{thm free}. Now, we discuss some examples of non-free LCD codes over $R_{e,q}$ which present many optimal LCD codes over $\mathbb{F}_q$ as the Gray images of these codes.
\begin{example}
Let $e=2,q=2\cdot 2+1=5,n=6$ and $R_{2,5}=\mathbb{F}_5[u]/\langle u^2-1\rangle$. Let $\mathcal{C}$ be a cyclic code of length $6$ over $R_{2,5}$. Now,
\begin{align*}
    x^6-1= (x + 1)(x + 4)(x^2 + x + 1)(x^2 + 4x + 1)\in \mathbb{F}_5[x].
\end{align*}
Let $g_1(x)=x+4$ and $g_2(x)=(x + 1)(x^2 + x + 1)=x^3 + 2x^2 + 2x + 1$. Since $g_1(x),g_2(x)$ are self-reciprocal polynomials, by Corollary \ref{cor rev} and Theorem \ref{thm free}, we have $\mathcal{C}=\langle \mu_1g_1(x)+\mu_2g_2(x)\rangle$ is a non-free LCD code of length $6$ over $R_{2,5}$. Let \[
   M=
  \left[ {\begin{array}{ccc}
   1 & 4 \\
   1 & 1
  \end{array} } \right]\in GL_2(\mathbb{F}_5),
\]
satisfying $MM^T=2I_2$. Then by Theorem \ref{thm gr lcd}, the Gray image $\varphi(\mathcal{C})$ is an LCD code with { parameters} $[12,8,4]$ and it is an optimal code according to \cite{Grassltable}.
\end{example}{}

\begin{example}
Let $e=3,q=3\cdot 2+1=7,n=3$ and $R_{3,7}=\mathbb{F}_7[u]/\langle u^3-1\rangle$. Now,
\begin{align*}
    x^3-1=(x + 3)(x + 5)(x + 6)\in \mathbb{F}_7[x].
\end{align*}
Let $g_1(x)=x + 6, g_2(x)=1$ and $g_3(x)=(x+2)(x+3)=x^2 + x + 1$. Then $\mathcal{C}=\langle \sum_{i=1}^3\mu_ig_i(x)\rangle$ is a cyclic code of length $3$ over $R_{3,7}$. Since $\gcd(4,7)=1$ and $g_i(x)$ is self-reciprocal polynomial for $i=1,2$, by Corollary \ref{cor rev} and Theorem \ref{thm free}, we say that $\mathcal{C}$ is a non-free LCD code. Again, let \[
   M=
  \left[ {\begin{array}{cccc}
   2 & 1 & 2 \\
   5 & 2 & 1 \\
   1 & 2 & 5 \\
  \end{array} } \right]\in GL_3(\mathbb{F}_7),
\]
satisfying $MM^T=2I_3$. Then by Theorem \ref{thm gr lcd}, the Gray image $\varphi(\mathcal{C})$ is an LCD code with { parameters} $[9,6,3]$, which is an optimal code by \cite{Grassltable}.
\end{example}

\begin{example}
Let $\mathcal{C}=\langle \mu_1g_1(x)+\mu_2g_2(x)\rangle$ be a non-free LCD code of length $n$ over $R_{2,q}$ where $\gcd(n,q)=1$ and $g_1(x),g_2(x)$ are self-reciprocal factors of $x^n-1$. Let \[
   M=
  \left[ {\begin{array}{ccc}
   1 & -1 \\
   1 & 1
  \end{array} } \right]\in GL_2(\mathbb{F}_q),
\]
satisfying $MM^T=2I_2$. Again, by Theorem \ref{thm gr lcd}, $\varphi(\mathcal{C})$ is an LCD code. In Table \ref{tab3}, we present generator polynomials $g_1(x),g_2(x)$ in { the} $3^{rd}$ and $4^{th}$ column, respectively. Also, their $\mathbb{F}_q$-images are included in { the} $5^{th}$ column which presents many optimal codes as per \cite{Grassltable}.
\end{example}{}

\begin{remark}
\begin{enumerate}
    \item In order to present Tables \ref{tab1}, \ref{tab2} { and} \ref{tab3} briefly, we write the vector consisting coefficients of the polynomial $g_i(x)$ in decreasing order. For example, we write the vector $(1,3,0,2,4)$ to represent the polynomial $x^4+3x^3+2x+4$.
    \item As we can see in Tables \ref{tab1} { and} \ref{tab3}, the codes { of parameters} $[16,11,4],[20,15,4],[24,19,4],[32,27,4],[12,8,3],$ $[12,7,4],[16,12,3],[26,22,3],[50,45,4],46,21,4]$ satisfy $n-k+1-d=2$ while { the} remaining codes satisfy $n-k+1-d=1$. Therefore, although these codes are not MDS (with respect to the singleton bound) { they are very close to be MDS, having good parameters.}
\end{enumerate}
\end{remark}

\section{Conclusion}
Here, we have obtained many MDS and optimal codes from the Gray images of both cyclic and LCD codes over $R_{e,q}$. To compute the parameters of codes, we have used the Magma computation system \cite{Magma}. Apart from the obtained parameters, one can further find more optimal codes using the derived results and { the} Gray map $\varphi$. The article justifies the fact that LCD codes over non-chain rings are a useful tool to obtain optimal LCD codes over $\mathbb{F}_q$. { We} hope our work would motivate researchers to study LCD codes over other non-chain rings to explore more good LCD codes in { the} future.

\section*{Acknowledgement}
The authors are thankful to the University Grants Commission (UGC), Govt. of India for financial supports under Ref. No. 20/12/2015(ii)EU-V dated 31/08/2016 and Indian Institute of Technology Patna for providing research facilities. E. Mart\'inez-Moro is partially funded by Spanish State Research Agency (AEI) under Grant PGC2018-096446-B-C21.\\
Also,  authors  would like to thank the anonymous referee(s) and the Editor for their valuable comments to improve the presentation of the manuscript.


\begin{table}
\caption{Some decomposition $\mu_i$ and their canonical components $s_i$}
\vspace{0.5cm}
\renewcommand{\arraystretch}{1.5}
\begin{center}

\begin{tabular}{|c|c|c|c|c|c|c|c|c|c|c||c|}
\hline
$e$& $q$ & $r$  & $\mu_i$ & $s_i$  \\
\hline
$2$ & $q$ & $ a_0+ua_1$ &$\mu_1=\frac{1+u}{2},\mu_2=\frac{1-u}{2}$ & $s_1=a_0+a_1,s_2=a_0-a_1$ \\
\hline

  &   &  &$\mu_1=6(u+5)(u+6)$ & $s_1=a_0+4a_1+2a_2$ \\
 3 &  7 & $ a_0+ua_1+u^2a_2$  &$\mu_2=3(u+3)(u+6)$ & $s_2=a_0+2a_1+4a_2$ \\
  &   &  &$ \mu_3=5(u+3)(u+5)$ & $ s_3=a_0+a_1+a_2$ \\
\hline

  &  &  &$\mu_1=3(u+10)(u+12)$ & $s_1=a_0+9a_1+3a_2$ \\
3  & 13  & $ a_0+ua_1+u^2a_2$  &$\mu_2=(u+4)(u+12)$ & $s_2=a_0+3a_1+9a_2$ \\
  &   &   &$ \mu_3=9(u+4)(u+10)$ & $ s_3=a_0+a_1+a_2$ \\
\hline

 &  &  &$\mu_1=(u+2)(u+3)(u+4)$ & $s_1=a_0+4a_1+a_2+4a_3$ \\
 4 & 5  &  $ a_0+ua_1+u^2a_2+u^3a_3$ &$\mu_2=2(u+1)(u+3)(u+4)$ & $s_2=a_0+3a_1+4a_2+2a_3$ \\
  &   &    &$ \mu_3=3(u+1)(u+2)(u+4)$ & $s_3=a_0+2a_1+4a_2+3a_3$ \\
    &   &    &$ \mu_4=4(u+1)(u+2)(u+3)$ & $s_4=a_0+a_1+a_2+a_3$ \\
\hline
  &   &  &$\mu_1=3(u+5)(u+8)(u+12) $ & $s_1=a_0+12a_1+a_2+12a_3 $ \\
 4 & 13  & $ a_0+ua_1+u^2a_2+u^3a_3$  &$ \mu_2=2(u+1)(u+8)(u+12)$ & $ s_2=a_0+8a_1+12a_2+5a_3$ \\
  &   &   &$ \mu_3=11(u+1)(u+5)(u+12)$ & $ s_3=a_0+5a_1+12a_2+8a_3$ \\
    &   &   &$ \mu_4=10(u+1)(u+5)(u+8)$ & $ s_4=a_0+a_1+a_2+a_3$ \\
\hline
 &  &  &$\mu_1=4(u+4)(u+13)(u+16) $ & $s_1=a_0+16a_1+a_2+16a_3, $ \\
 4 & 17  & $ a_0+ua_1+u^2a_2+u^3a_3$  &$ \mu_2=16(u+1)(u+13)(u+16)$ & $ s_2=a_0+13a_1+16a_2+4a_2$ \\
  &   &   &$ \mu_3=(u+1)(u+4)(u+16)$ & $ s_3=a_0+4a_1+16a_2+13a_3$ \\
    &   &   &$ \mu_4=13(u+1)(u+4)(u+13)$ & $ s_4=a_0+a_1+a_2+a_3$ \\
\hline

\end{tabular}\label{tabA}
\end{center}
\end{table}



\begin{table}
\caption{Optimal linear as the Gray images of cyclic codes over $R_{2,q}=\mathbb{F}_q[u]/\langle u^2-1\rangle$}
\vspace{0.5cm}
\renewcommand{\arraystretch}{1.5}
\begin{center}
\begin{tabular}{|c|c|c|c|c|c|c|c|c|c|c|c|}

\hline
$q$ &$n$ & $g_1(x)$ & $g_2(x)$ & $[n,k,d]$ & Remark\\
\hline
$5$ &$4$ & $(1,4)$ & $(1,2,2)$ & $[8,5,3]$ & Optimal\\

$5$ &$4$ & $(1,4,3)$ & $(1,2,2)$ & $[8,4,4]$ & Optimal\\

$5$ &$4$ & $(1,3,4,2)$ & $(1,4,1,4)$ & $[8,2,6]$ & Optimal\\

$5$ &$5$ & $(1,3,1)$ & $(1,4)$ & $[10,7,3]$ & Optimal\\

$5$ &$5$ & $(1,4)$ & $(1,2,3,4)$ & $[10,6,4]$ & Optimal\\

$5$ &$8$ & $(1,3,2,1)$ & $(1,4,3)$ & $[16,11,4]$ & Optimal\\

$5$ &$10$ & $(1,4)$ & $(1,3,0,2,4)$ & $[20,15,4]$ & Optimal\\

$7$ &$6$ & $(1,5,6)$ & $(1,2)$ & $[12,9,3]$ & Optimal\\

$7$ &$7$ & $(1,5,1)$ & $(1,6)$ & $[14,11,3]$ & Optimal\\

$7$ &$7$ & $(1,4,3,6)$ & $(1,6)$ & $[14,10,4]$ & Optimal\\

$7$ &$8$ & $(1,2,5,6)$ & $(1,6)$ & $[16,12,4]$ & Optimal\\

$7$ &$12$ & $(1,6)$ & $(1,5,0,5,6)$ & $[24,19,4]$ & BKLC\\

$7$ &$14$ & $(1,5,0,2,6)$ & $(1,6)$ & $[28,24,4]$ & BKLC\\

$7$ &$14$ & $(1,6)$ & $(1,3,5,4,1)$ & $[32,27,4]$ & Optimal\\

$3$ &$3$ & $(1,1,1)$ & $(1,2)$ & $[6,3,3]$ & Optimal\\

$3$ &$6$ & $(1,2,2,1)$ & $(1,2)$ & $[12,8,3]$ & Optimal\\

$3$ &$6$ & $(1,1,0,2,2)$ & $(1,2)$ & $[12,7,4]$ & Optimal\\

$3$ &$8$ & $(1,2)$ & $(1,1,0,1,2)$ & $[16,11,4]$ & Optimal\\

$3$ &$8$ & $(1,2)$ & $(1,2,0,2)$ & $[16,12,3]$ & Optimal\\

$3$ &$13$ & $(1,0,2,2)$ & $(1,2)$ & $[26,22,3]$ & Optimal\\

$9$ &$8$ & $(1,w,w^5,2)$ & $(1,w^6)$ & $[16,12,4]$ & Optimal\\

$9$ &$10$ & $(1,w^2,w^6,2)$ & $(1,2)$ & $[20,16,4]$ & Optimal\\

\hline
\end{tabular}\label{tab1}
\end{center}
\end{table}



\begin{table}
\caption{linear MDS codes as the Gray images of cyclic codes over $R_{2,q}=\mathbb{F}_q[u]/\langle u^2-1\rangle$}
\vspace{0.5cm}
\renewcommand{\arraystretch}{1.5}
\begin{center}
\begin{tabular}{|c|c|c|c|c|c|c|c|c|c|c|c|}

\hline
$q$ &$n$ & $g_1(x)$ & $g_2(x)$ & $[n,k,d]$ & Remark\\
\hline
$7$ &$3$ & $(1,3)$ & $(1,5)$ & $[6,4,3]$ & MDS\\

$7$ &$3$ & $(1,3)$ & $(1,4,2)$ & $[6,3,4]$ & MDS\\

$7$ &$3$ & $(1,1,1)$ & $(1,4,2)$ & $[6,2,5]$ & MDS\\

$11$ &$5$ & $(1,6)$ & $(1,8)$ & $[10,8,3]$ & MDS\\

$11$ &$5$ & $(1,7,3)$ & $(1,2,9)$ & $[10,6,5]$ & MDS\\

$19$ &$9$ & $(1,3)$ & $(1,8)$ & $[18,6,3]$ & MDS\\

$23$ &$11$ & $(1,5)$ & $(1,7)$ & $[22,20,3]$ & MDS\\

\hline
\end{tabular}\label{tab2}
\end{center}
\end{table}
%


\begin{table}
\caption{Optimal LCD codes as the Gray images of non-free LCD codes over $R_{2,q}=\mathbb{F}_q[u]/\langle u^2-1\rangle$}
\vspace{0.5cm}
\renewcommand{\arraystretch}{1.5}
\begin{center}
\begin{tabular}{|c|c|c|c|c|c|c|c|c|c|c|c|}

\hline
$q$ &$n$ & $g_1(x)$ & $g_2(x)$ & $[n,k,d]$ & Remark\\
\hline
$5$ &$3$ & $(1,1,1)$ & $(1,4)$ & $[6,3,4]$ & MDS\\

$5$ &$6$ & $(1,4)$ & $(1,2,2,1)$ & $[12,8,4]$ & Optimal\\

$5$ &$12$ & $(1,1,2,1,1)$ & $(1,4)$ & $[24,19,4]$ & Optimal\\

$5$ &$24$ & $(1,2,4,4,2,1)$ & $(1,4)$ & $[48,42,4]$ & Optimal\\

$7$ &$3$ & $(1,6)$ & $(1,1,1)$ & $[6,3,4]$ & MDS\\

$7$ &$4$ & $(1,5,2,6)$ & $(1,1)$ & $[12,8,4]$ & Optimal\\

$7$ &$8$ & $(1,4,4,1)$ & $(1,1)$ & $[16,12,4]$ & Optimal\\

$7$ &$25$ & $(1,6)$ & $(1,2,4,2,1)$ & $[50,45,4]$ & Optimal\\

$5$ &$13$ & $(1,1,4,1,1)$ & $(1,4)$ & $[26,21,4]$ & Optimal\\

$3$ &$4$ & $(1,2,1,2)$ & $(1,2)$ & $[8,4,4]$ & Optimal\\

\hline
\end{tabular}\label{tab3}
\end{center}
\end{table}
%

\end{document}